%% file: main.tex
\documentclass[11pt]{article}
\input{preamble}

\begin{document}

\allowdisplaybreaks

\title{Nearly-Optimal Private Selection via Gaussian Mechanism}
\date{\today}

\author{Ethan Leeman \\Google Research\\ \texttt{\small ethanleeman@google.com} \and Pasin Manurangsi \\Google Research\\ \texttt{\small pasin@google.com}}

\maketitle

\begin{abstract}
%
Steinke~\cite{DPorg-open-problem-selection} recently asked the following intriguing open question: Can we solve the differentially private selection problem with nearly-optimal error by only (adaptively) invoking Gaussian mechanism on low-sensitivity queries? We resolve this question positively. In particular, for a candidate set $\calY$, we achieve error guarantee of $\tO{\log |\calY|}$, which is within a factor of $(\log \log |\calY|)^{O(1)}$ of the exponential mechanism~\cite{McSherryT07}. This improves on Steinke's mechanism which achieves an error of $O(\log^{3/2} |\calY|)$.
\end{abstract}

\section{Introduction}

Differential privacy (DP)~\cite{dwork-calibrating} has become one of the most widely used notion of privacy, partly due to its mathematically rigorous protection for the user's data. In the original work that proposed DP~\cite{dwork-calibrating}, Dwork et al. also provided a mechanism that satisfies DP: If the desired output has low sensitivity, then we can simply add an appropriately-calibrated Laplace noise. Later on, other noise distribution--such as the Gaussian distribution~\cite{DworkKMMN06}--has also been studied. These so-called Laplace and Gaussian mechanisms remain a staple of DP, both in theory and practice.

Meanwhile, another widely-used mechanism is the \emph{Exponential Mechanism}~\cite{McSherryT07}, which can be easily described through the \emph{Selection} problem, defined below.

\begin{definition}[Selection]
In the Selection problem, there is a candidate set $\calY$, and for each candidate $y \in \calY$ there is a sensitivity-1 loss function $\ell_y: \calX^* \to \R$. 
Given an input $X \in \calX^*$, the goal of a selection algorithm is to output $y_{\out}$ that approximately minimizes the loss $\ell_y(X)$. We say that an algorithm $M: \calX^* \to \calY$ solves Selection with expected error $\nu$ if
\begin{align*}
\E_{y_{\out} \sim M(X)}[\ell_{y_{\out}}(X) - \min_{y \in \calY} \ell_y(X)] \leq \nu.
\end{align*}
Similarly, we say that an algorithm solves Selection with error $\nu$ with probability $1 - \beta$ iff
\begin{align*}
\Pr_{y_{\out} \sim M(X)}[\ell_{y_\out}(X) - \min_{y \in \calY} \ell_y(X) \leq \nu] \geq 1 - \beta.
\end{align*}
When there is no ambiguity, we simply write $\ell_y$ instead of $\ell_y(X)$ for brevity.
\end{definition}

The $\eps$-DP Exponential Mechanism solves Selection with expected error $O(\log |\calY| / \eps)$. Since Selection can be used to model a variety of tasks, such as convex optimization~\cite{BassilyST14}, combinatorial optimization~\cite{GuptaLMRT10} and PAC learning~\cite{BeimelNS16}. Exponential mechanism immediately yields DP algorithms for these tasks, often with nearly tight error guarantees.

Given the importance of both types of mechanism, it is natural to ask whether we can implement the Exponential Mechanism using only noise addition mechanisms. Steinke~\cite{DPorg-open-problem-selection} recently formalized this problem using the Gaussian mechanism. To further elaborate on this, it is best to focus on the notion of zero-concentrated differential privacy (zCDP)~\cite{DworkR16,BunS16}. Any $\eps$-DP satisfies $\frac{\eps^2}{2}$-zCDP; thus, the aforementioned Exponential Mechanism is an $\rho$-zCDP mechanism which solves Selection with expected error $O(\log |\calY| / \sqrt{\rho})$. Meanwhile, adding Gaussian noise $\calN\left(0, \frac{1}{2\rho}\right)$ to any sensitivity-1 query satisfies $\rho$-zCDP. Furthermore, the composition theorem for zCDP is to simply add up to $\rho$ values\footnote{In particular, this works even in the fully adaptive composition (aka privacy filter) setting~\cite{WhitehouseR0023}.}. As a result, one may formalize the model as follows\footnote{Strictly speaking, Steinke~\cite{DPorg-open-problem-selection} proposed a model where the queries have equal budget. However, it is simple to see that the two models are equivalent; see \Cref{app:eqaul-v-unequal}.}:
\begin{definition}[Gaussian Mechanism with Budget~\cite{DPorg-open-problem-selection}] \label{def:model-with-budget}
In the \emph{Gaussian mechanism with budget} model, the algorithm is given a finite budget $\rho > 0$ and can select the number of rounds $M$. At the $i$-th iteration, the algorithm can issue a sensitivity-1 query $q_i: \calX^* \to \R$ and the budget $\rho_i$, and it will receive back $q_i(X) + Z_i$ where $Z_i \sim \calN\left(0, \frac{1}{2\rho_i}\right)$. Here we enforce that $\sum_{i=1}^M \rho_i \leq \rho$. 
\end{definition}

\cite{DPorg-open-problem-selection} asked whether one can solve Selection with expected error $O(\log |\calY| / \sqrt{\rho})$ (similar to the Exponential Mechanism) in this model. Note that the trivial baseline--which queries $\ell_y(X)$ for all $y \in \calY$ using equal budget--results in an error of $\tO{\sqrt{|\calY|} / \sqrt{\rho}}$, significantly higher than that of the Exponential Mechanism. Perhaps surprisingly, Steinke~\cite{DPorg-open-problem-selection} gave a simple algorithm that significantly improves upon this and achieves an expected error of $O\left(\frac{\log^{3/2} |\calY|}{\sqrt{\rho}}\right)$. (See \Cref{sec:prelim-bintree} for more detail.) Nevertheless, there is still a gap of $O(\sqrt{\log |\calY|})$ between this bound and the error of the Exponential Mechanism. Thus, Steinke~\cite{DPorg-open-problem-selection} asks whether this gap can be closed. 

\subsection{Our Contribution}

The main contribution of our work is to answer this positively, up to a lower-order term. In particular, we give an algorithm with error $\O{\frac{\log |\calY| \cdot (\log \log |\calY|)^{O(1)}}{\sqrt{\rho}}}$, as stated below.
\begin{theorem}[Informal; See \Cref{thm:main-exp}]
There is an algorithm in the Gaussian mechanism with budget model that solves Selection with expected error $\O{\frac{\log |\calY| \cdot (\log \log |\calY|)^{11}}{\sqrt{\rho}}}$.
\end{theorem}

At a high-level, our algorithm exploits the following observation: When there is a large \emph{gap} between the smallest and second-smallest losses, the problem becomes significantly easier. In fact, when the gap is $\tO{\sqrt{\log |\calY|}/\sqrt{\rho}}$, Steinke's algorithm can already find the optimal index with high probability. Given this, our algorithm then proceeds in three steps. First, use subsampling to produce multiple instances (with different sampling rates). Secondly, recursively solve Selection on these instances but with the loss set in such a way that we prefer an instance with a large gap. Finally, we run Steinke's algorithm on the instance with large gap. The main challenge in the analysis is to ensure that, by carefully selecting the different sampling rates, at least one of the subsampled instances has both a large gap and a small optimal loss.

\subsection{Discussion and Open Questions}

An interesting open question here is to remove the lower order term (of $(\log \log |\calY|)^{O(1)}$) and completely achieve $O\left(\frac{1}{\sqrt{\rho}} \cdot \log |\calY|\right)$ error, which would exactly match the guarantee from the Exponential Mechanism. We remark that below we did not try to optimize the exponent of 11 in our $(\log \log |\calY|)^{O(1)}$ term, and indeed we suspect that it might be possible to make exponent arbitrarily close to two by slightly modifying the algorithm and the proof. Nevertheless, we do not see how the current technique can get rid of such a term completely. 

Another interesting direction is to consider the model with \emph{Laplace} mechanism, instead of Gaussian mechanism. Steinke's mechanism~\cite{DPorg-open-problem-selection} gives an error of $O_{\eps}(\log^2 |\calY|)$. Unfortunately, our technique does not seem to achieve any improvement in this setting since Laplace noise is not as concentrated as Gaussian; thus, one needs a much larger gap when working with Laplace noise.

Lastly, the algorithm we describe can only be implemented to solve Selection when the candidate set $|\calY|$ is finite. Meanwhile, the exponential mechanism can be well-defined on probability spaces with infinite sample spaces. It would be interesting to discover ways to use low-sensitivity queries to implement Selection in this setting. We believe it would require fundamentally different uses of the low-sensitivity queries, as opposed to the tree-based method used here and in ~\cite{DPorg-open-problem-selection}.

\section{Preliminaries}

Throughout, we use $\log$ to denote base-2 logarithm.

The sensitivity of a function $g: \calX^* \to \R$ is $\Delta(g) := \max_{\bX, \bX'} |g(\bX) - g(\bX')|$ where the maximum is over all neighboring inputs $\bX, \bX'$. All results stated below work for any neighboring notion.



\subsection{The Binary Tree Algorithm}
\label{sec:prelim-bintree}

We recall the ``binary tree'' algorithm proposed in \cite{DPorg-open-problem-selection}. Roughly speaking, the algorithm constructs a balanced full binary tree with the candidates being the leaves. The algorithm then starts at the root and, at each step, uses a single low-sensitivity query to determine whether the minimum value of the left subtree or the right subtree is larger; it then traverses accordingly. The algorithm is described more formally in \Cref{alg:binary-tree}. The guarantee shown in \cite{DPorg-open-problem-selection} is as follows:

\begin{lemma}[\cite{DPorg-open-problem-selection}] \label{lem:bin-search-exp}
For any $\rho > 0$, \Cref{alg:binary-tree} solves Selection with expected error $O\left(\frac{\log^{3/2} |\calY|}{\sqrt{\rho}}\right)$. 
\end{lemma}

\begin{algorithm}[h!]
\caption{\bintree$\left((\ell_y)_{y \in \calY}; \rho\right)$~\cite{DPorg-open-problem-selection}}
\label{alg:binary-tree}
\begin{algorithmic}[1]
\Require{Loss functions $(\ell_y)_{y \in \calY}$,~} \textbf{Parameter: } Privacy budget $\rho$
\Ensure{Index $y_{\out} \in \calY$}
\State{$\calC \gets \calY$} \Comment{$\calC$ is the current candidate set}
\State{$K \gets \lceil \log |\calY| \rceil$}
\While{$|\calC| > 1$}
\State Partition $\calC$ into $\calC_1 \cup \calC_2$ where the sizes of $\calC_1, \calC_2$ differ by at most one.
\State Query $q := \frac{1}{2}\left(\min_{y_1 \in \calC_1} \ell_{y_1}(X) - \min_{y_2 \in \calC_2} \ell_{y_2}(X)\right)$ with budget $\frac{\rho}{K}$; let the answer be $\tilde{q}$
\If{$\tq > 0$}
\State $\calC \gets \calC_2$
\Else
\State $\calC \gets \calC_1$
\EndIf
\EndWhile
\State \Return the only index in $\calC$
\end{algorithmic}
\end{algorithm}

\section{Nearly Optimal Algorithm for Selection}

The main result of our work is an improvement over Steinke's algorithm (\Cref{lem:bin-search-exp}). Namely, we devise an algorithm with expected error $\tO{\frac{1}{\sqrt{\rho}} \cdot \log |\calY|}$, as stated formally below.

\begin{theorem}[Main Theorem] \label{thm:main-exp}
For any $\rho > 0$, \Cref{alg:combined} solves Selection with expected error $O\left(\frac{1}{\sqrt{\rho}} \cdot \log |\calY| \cdot (\log \log |\calY|)^{11}\right)$.
\end{theorem}

The remainder of the section is devoted to the proof of \Cref{thm:main-exp}.

\subsection{Binary Tree, Revisited: Gap Instances Are Easy}

We start by observing that, in many cases, the binary tree algorithm is already doing much better than \Cref{lem:bin-search-exp} is suggesting. In particular, when there is a ``gap'' between the smallest and second smallest loss values, the algorithm already does well. To state this, let us first formalize the ``gap''.

\begin{definition}[Minimum Index]
We say that $y^*$ is a \emph{minimum index} of $(\ell_y)_{y \in \calY}$ if $\ell_{y^*}(X) = \min_{y \in \calY} \ell_y(X)$.
\end{definition}
\begin{definition}[Gap] \label{def:gap}
For any non-empty set $\calY$, the \emph{gap} of an instance $(\ell_y)_{y \in \calY}$ is defined as
$$\gap\left((\ell_y)_{y \in \calY}\right) := \min_{y \in \calY \setminus \{y^*\}} \ell_y(X) - \ell_{y^*}(X),$$
where $y^*$ is any minimum index of $\calY$. When $|\calY| = 1$, we define the gap to be $\infty$.
\end{definition}

Our observation here is that if an instance have gap at least $\tilde{\Theta}\left(\sqrt{\log |\calY|} / \sqrt{\rho}\right)$, then we output the minimum index w.h.p., as stated more formally below.

\begin{lemma}[Binary Tree for Gap Instances] \label{lem:gap-bin-search}
For any $\rho > 0, \beta \in (0, 0.5)$, let $\tau(|\calY|, \rho, \beta) := \frac{2\sqrt{\lceil \log|\calY| \rceil \cdot \left(\log\left(\frac{\lceil \log|\calY| \rceil}{\beta}\right)\right)}}{\sqrt{\rho}}$. 
If an instance $(\ell_y)_{y \in \calY}$ has gap at least $\tau(|\calY|, \rho, \beta)$, then \Cref{alg:binary-tree} outputs the minimum index $y^*$ with probability at least $1 - \beta$.
\end{lemma}

\begin{proof}
By standard concentration of Gaussian and the union bound, we have that $|\tilde{q} - q| < \frac{1}{\sqrt{(\rho/K)}} \cdot \sqrt{\log(K / \beta)} = \frac{\tau(|\calY|, \rho, \beta)}{2}$ for all rounds with probability at least $1 - \beta$. When this holds, since the instance has gap at least $\tau(|\calY|, \rho, \beta)$, we will always pick $\calC_i$ that contains the minimum index $y^*$; that is, we output $y^*$ in the end.
\end{proof}

\subsection{Our New Recursive Algorithm and Its Guarantee}

In this subsection, we present our main contribution, which is an algorithm that with high probability achieves an error bound of $\tO{\frac{1}{\sqrt{\rho}} \cdot \log |\calY|}$. The exact bound is stated below.

\begin{theorem}[Main Theorem] \label{thm:main}
For any $\rho > 0, \beta \in (0, 0.001)$, \Cref{alg:final} solves Selection with error $O\left(\frac{1}{\sqrt{\rho}} \cdot \log |\calY| \cdot (\log \log |\calY|)^{10} \cdot \log\left(\frac{\log |\calY|}{\beta}\right)\right)$ with probability at least $1 - \beta$.
\end{theorem}

Before we present the algorithm, let us describe the intuition behind the algorithm. The high-level description of the algorithm is very simple: (i) subsampling ``many'' subsets and (ii) recursively call the algorithm to select the subset with (approximately) largest gap, and then (iii) call \bintree~(using the guarantee in \Cref{lem:gap-bin-search}). The main argument is regarding how should we sample and how ``many'' subsets do we need to get a gap of $\tO{\sqrt{\log |\calY|}}$ as required in \Cref{lem:gap-bin-search}. For convenience of discussion, suppose that $|\calY| = 2^K$ where $K \in \mathbb{N}$ is a perfect square. Intuitively, a ``worst'' case here is the following: For all $i \in [K]$, there are $2^i$ candidates whose losses are within $\tO{i}$ of the optimal. The reason here is that, if the losses are ``more packed'', then all candidates would result in an error less than $\tO{K}$, which would satisfy our desired bound. On the other hand, if the losses are ``less packed'', then the gaps are larger and thus it is easier for our steps (i) and (ii). Now, coming back to this instance, if we sample $2^{K - \sqrt{K}}$ candidates from $\calY$, then with probability $\Theta(2^{-\sqrt{K}})$, we will select the optimal candidate and none of the other $2^{\sqrt{K}}$ candidates that are within $\tO{\sqrt{K}}$-close to it. Such a set yields $\tO{\sqrt{K}}$ gap required for \Cref{lem:gap-bin-search}. Moreover, since the probability of sampling such a subset is $\Theta(2^{-\sqrt{K}})$, it suffices for us to sample $O(2^{\sqrt{K}})$ such sets. This means that, in the recursive step (ii), the ``error'' in the selection is $\tO{\sqrt{K}}$ which can be absorbed by the gap if we select the parameters carefully. 

While the above presents the overview for the ``worst'' instance, in general the instance can be a lot more complicated; e.g. for some $i \in [K]$, there might be more or fewer than $2^i$ candidates whose losses are within $\tO{i}$. For example, if there are many more candidates (say $2^{\omega(\sqrt{K})}$) whose loss are very close to the optimum, then sampling as above will not be sufficient. Thus, in the main algorithm, we also randomly sample the sampling probability (i.e. subset size) by taking them to be a random power of two. We can show that, at least one such power of two results in a large gap with probability at least $2^{-\Omega(\sqrt{K})}$ (\Cref{lem:prob-good-score-one-subset}). We also need to incorporate the minimum loss value in the subset into the loss function in the recursive step (ii) (see \Cref{line:loss-recur-depth} in \Cref{alg:final}); this is to prevent us from selecting a subset with a large gap but with a large optimal loss value. Our entire algorithm is given in \Cref{alg:final}.

\begin{algorithm}[h!]
\caption{\recurgap$\left((\ell_y)_{y \in \calY}; \rho, \beta\right)$}
\label{alg:final}
\begin{algorithmic}[1]
\Require{Loss functions $(\ell_y)_{y \in \calY}$,~} \textbf{Parameter: } Privacy budget $\rho$, failure probability $\beta$
\Ensure{Index $y_{\out} \in \calY$}
\If{$|\calY| \leq 2^{1000}$ or $\beta \leq 2^{-K}$}
\State \Return \bintree$\left((\ell_y)_{y \in \calY}; \rho\right)$
\EndIf
\State $K \gets \lceil \log |\calY| \rceil$
\State $T \gets \left\lceil 2^{3\sqrt{K}-1}\right\rceil$.
\State $\xi(K, \rho, \beta) \gets \frac{1000}{\sqrt{\rho}} (1 + \log K)^{10} \log(1000(K+1)/\beta)$ 
\label{line:gap-def}
\State $y^* \gets$ minimum index of $\calY$ \Comment{Tie broken arbitrarily}
\For{$t \in [T]$}
\State Sample uniformly $k_t \sim [K]$.
\State $\calS_t \gets$ a random subset of $\calY$ of size $2^{K - k_t}$.
\State $\tell_t \gets \frac{1}{2} \cdot \max\left\{\min_{y \in \calS_t} \ell_y - \ell_{y^*} - (K + \sqrt{K}) \cdot \xi(K, \rho, \beta), -\gap\left((\ell_y)_{y \in \calS_t}\right)\right\}$ \label{line:loss-recur-depth} 
\EndFor
\State $t_{\out} \gets \recurgap\left((\tell_t)_{t \in [T]}; \frac{4\rho}{5}, \frac{4\beta}{5}\right)$
\State $y_{\out} \gets \bintree\left((\ell_y)_{y \in \calS_{t_{\out}}}; \frac{\rho}{5}\right)$
\State \Return $y_{\out}$
\end{algorithmic}
\end{algorithm}

\subsubsection{Proof of \Cref{thm:main}}

We will now formalize the above intuition \Cref{thm:main}.
It is simple to check that the created loss functions $\tell_t$ have sensitivity at most one. As such, our algorithm is valid in the model. We next proceed to argue its accuracy guarantee.

Recall the definition of $\xi$ from \Cref{line:gap-def} of \Cref{alg:final}. Throughout the remainder of the proof, we also let
\begin{align} \label{eq:err-def-high-prob}
\gamma(m, \rho, \beta) = 2 \cdot \lceil \log m \rceil \cdot \xi(\lceil \log m \rceil, \rho, \beta).
\end{align}

We will use the following two inequalities. Here $K$ and $T$ are as defined in \Cref{alg:final}.

\begin{lemma} \label{lem:first-ineq}
If $K \geq 1000$ and $\beta \in (0, 0.001)$, then
$K \cdot \xi\left(K, \rho, \beta\right) + 2\gamma\left(T, \frac{4\rho}{5}, \frac{4\beta}{5}\right) \leq \gamma\left(|\calY|, \rho, \beta\right)$
\end{lemma}

\begin{lemma} \label{lem:second-ineq}
If $K \geq 1000$ and $\beta \in (0, 0.001)$, then
$\sqrt{K} \cdot \xi\left(K, \rho, \beta\right) - 2\gamma\left(T, \frac{4\rho}{5}, \frac{4\beta}{5}\right) \geq \tau\left(|\calY|, \frac{\rho}{5}, \frac{\beta}{10}\right)$
\end{lemma}

The proof of these lemmas (which are mostly tedious calculations) are deferred to \Cref{app:ineq-proofs}. For the remainder of the proof, we write $\xi$ as a shorthand for $\xi(K, \rho, \beta)$.

 With these lemmas stated, we are now ready to prove \Cref{thm:main}. 
We prove by induction on $|\calY|$ that \Cref{alg:final} solves Selection with error at most $\gamma(|\calY|, \rho, \beta)$ with probability $1 - \beta$.

\paragraph{Base Case.}
For $|\calY| \leq 2^{1000}$, we simply run \bintree. Similar to the proof of \Cref{lem:gap-bin-search}, with probability $1 - \beta$, the noises added in every step in \Cref{alg:binary-tree} is at most $\sqrt{\frac{K \log (K/\beta)}{\rho}}$. When this is the case, the error is at most $2K \cdot \sqrt{\frac{K \log (K/\beta)}{\rho}}$. For $K \leq 1000$, it is simple to verify that this is at most $\gamma(|\calY|, \rho, \beta)$.

\paragraph{Inductive Step.} 
Assume that the guarantee holds for all $|\calY| < m$ where $m \geq 2^{1000}$. We will show that this also holds when $|\calY| = m$. First, note that if $\beta \leq 2^{-K}$, then we have that $\gamma(|\calY|, \rho, \beta) \geq 2K \cdot \sqrt{\frac{K \log (K/\beta)}{\rho}}$. Thus, as argued above, \Cref{alg:binary-tree} suffices to handle this case. Henceforth, we may assume that $\beta > 2^{-K}$.

We start with the following lemma, which shows that our random sampling procedure results in a ``good'' set with probability at least $2^{-\Omega(\sqrt{K})}$.

\begin{lemma} \label{lem:prob-good-score-one-subset}
For each $t \in [T]$, $\Pr\left[2 \tell_t \leq - \sqrt{K} \cdot \xi\right] \geq \frac{1}{2^{2\sqrt{K}}}.$
\end{lemma}

\begin{proof}
Let $y^{(1)}, \dots, y^{(2^K)}$ be the elements of $\calY$ ordered in non-decreasing losses. (Note that $y^{(1)} = y^*$.) Let
\begin{align*}
i^* := \argmin_{i \in \{0, \dots, K\}} \ell_{y^{(2^i)}} - i \cdot \xi.
\end{align*}
Notice that, by our choice of $i^*$, we have
\begin{align*}
\ell_{y^{(2^{i^*})}} - \ell_{y^{(1)}} = i^* \cdot \xi + \left(\ell_{y^{(2^{i^*})}} - i^* \cdot \xi\right) - \left(\ell_{y^{(1)}} - 0 \cdot \xi\right) \leq i^* \cdot \xi \leq K \cdot \xi. 
\end{align*}
Rearranging this yields
\begin{align} \label{eq:gapsmallest-and-istar}
\ell_{y^{(2^{i^*})}} - \ell_{y^{(1)}} - (K + \sqrt{K}) \cdot \xi  \leq -\sqrt{K} \cdot \xi
\end{align}

Next, let $k^* := i^* + \lceil \sqrt{K} \rceil$. Consider two cases:
\begin{itemize}
\item Case I: $k^* \geq K$. In this case, consider the event $k_t = K$ and the only element of $\calS_t$ comes from $\{y^{(1)}, \dots, y^{(2^{i^*})}\}$. This event occurs with probability
\begin{align*}
\frac{1}{K} \cdot \frac{2^{i^*}}{2^K} \geq \frac{1}{K} \cdot \frac{1}{2^{\sqrt{K} + 1}} \geq \frac{1}{2^{2\sqrt{K}}},
\end{align*}
where the first inequality follows from $i^* + \lceil \sqrt{K} \rceil = k^* \geq K$ and the second inequality comes from the lower bound on $K$. 

When these event occurs, we have
\begin{align*}
2\tell_t \leq \ell_{y^{(2^{i^*})}} - \ell_{y^{(1)}} - (K + \sqrt{K}) \cdot \xi \overset{\eqref{eq:gapsmallest-and-istar}}{\leq} - \sqrt{K} \cdot \xi,
\end{align*}
where, in the first inequality, we also use that fact that $\gap\left((\ell_y)_{y \in \calS_t}\right) = \infty$ since $|\calS_t| = 1$.

\item Case II: $k^* \leq K$. In this case, consider the following two events:
\begin{itemize}
\item $\calE$: $|\{y^{(1)}, \dots, y^{2^{(i^*)}}\} \cap \calS_t| = 1$.
\item $\calE'$: $\{y^{(2^{(i^*)}+1)}, \dots, y^{2^{k^*}}\} \cap \calS_t = \emptyset$.
\end{itemize}
Notice that, when both of these events hold, we have 
\begin{align*}
\gap(\calS_t) &\geq \ell_{y^{(2^{k^*}+1)}} - \ell_{y^{(2^{i^*})}}  \\
&= (k^* \cdot \xi - i^* \cdot \xi) + \left(\ell_{y^{(2^{k^*})}} - k^* \cdot \xi\right) - \left(\ell_{y^{(2^{i^*})}} - i^* \cdot \xi\right) \\
&\geq (k^* \cdot \xi - i^* \cdot \xi) \\
&\geq \sqrt{K} \cdot \xi,
\end{align*}
where the second inequality is due to our choice of $i^*$ and the third is due to the choice of $k^*$.

Recall also \Cref{eq:gapsmallest-and-istar}.
As a result, when these events occur, we have $2\tell_t \leq - \sqrt{K} \cdot \xi$.

Finally, notice that the probability that these events occur can be bounded as follows:
\begin{align*}
\Pr[\calE \wedge \calE'] &\geq \Pr[k_t = k^*] \cdot \Pr[\calE \wedge \calE' \mid k_t = k^*] \\
&= \frac{1}{K} \cdot \frac{2^{i^*} \cdot \binom{2^K-2^{k^*}}{2^{K-k^*}-1}}{\binom{2^K}{2^{K-k^*}}} \\
&= \frac{1}{K} \cdot \frac{2^{i^*} \cdot \binom{2^K-2^{k^*}}{2^{K-k^*}-1}}{\frac{2^K}{2^{K-k^*}} \cdot \binom{2^K-1}{2^{K-k^*}-1}} \\
&= \frac{1}{K} \cdot \frac{1}{2^{k^* - i^*}} \cdot \prod_{j=0}^{2^{K-k^*}-2} \frac{2^K-2^{k^*}-j}{2^K-1-j} \\
&= \frac{1}{K} \cdot \frac{1}{2^{k^* - i^*}} \cdot \prod_{j=0}^{2^{K-k^*}-2} \left(1 - \frac{2^{k^*}-1}{2^K -1 - j} \right) \\
&\geq \frac{1}{K} \cdot \frac{1}{2^{\sqrt{K} + 1}} \cdot \left(1 - \frac{2^{k^*}-1}{2^K - 2^{K-k^*}+1} \right)^{2^{K-k^*}} \\
&\geq \frac{1}{K} \cdot \frac{1}{2^{\sqrt{K} + 1}} \cdot \left(1 - \frac{2^{k^*}}{2^K - 2^{K-k^*}} \right)^{2^{K-k^*}}. \\
\end{align*}

The last term on the RHS can be bounded using the Bernoulli's inequality as follows. (Recall that $k^* \leq K - 1$ in this case.)
\begin{align*}
\left(1 - \frac{2^{k^*}}{2^K - 2^{K-k^*}} \right)^{2^{K-k^*}} &= \left(\left(1 - \frac{2^{k^*}}{2^K - 2^{K-k^*}} \right)^{2^{K-k^*-1}}\right)^2 \\
&\geq \left(1 - \frac{2^{K-1}}{2^K - 2^{K-k^*}}\right)^2 \\
&= \left(1 - \frac{1}{2 - 2^{1-k^*}}\right)^2 \\
&\geq \left(1 - \frac{1}{1.5}\right)^2 \geq 0.1,
\end{align*}
where we use the Bernoulli's inequality in the first inequality above, and the second inequality follows from $\sqrt{K} > 1$.

Combining the above two inequality, we get 
\begin{align*}
\Pr[\calE \wedge \calE'] \geq \frac{1}{K} \cdot \frac{1}{2^{\sqrt{K} + 1}} \cdot 0.1 \geq \frac{1}{2^{2\sqrt{K}}}.
\end{align*}
\end{itemize}
Thus, in both cases, the desired inequality holds.
\end{proof}

With the above lemma ready, we are now ready to finish the accuracy proof.

\paragraph{Good Events.}
Consider the following ``good'' events:
\begin{itemize}
\item $\calE_1$: $\min_{t \in [T]} 2\tell_t \leq - \sqrt{K} \cdot \xi$
\item $\calE_2$: $\tell_{t_{\out}} \leq \gamma\left(T, \frac{4\rho}{5}, \frac{4\beta}{5}\right) + \min_{t \in [T]} \tell_t$
\item $\calE_3: \ell_{y_{\out}} = \min_{y \in \calS_{t_{\out}}} \ell_y$
\end{itemize}

\paragraph{From Good Events to Accuracy.}
Before we bound the probability that these event occurs, let us show argue that, when these good events occur, we get the desired accuracy. To do this, first note that $\calE_1, \calE_2$ together implies
\begin{align}
&2\gamma\left(T, \frac{4\rho}{5}, \frac{4\beta}{5}\right) - \sqrt{K} \cdot \xi \nonumber \\
&\geq 2\tell_{t_{\out}} = \max\left\{\min_{y \in \calS_{t_{\out}}} \ell_y - \ell_{y^*} - (K + \sqrt{K}) \cdot \xi(K, \rho, \beta), -\gap\left((\ell_y)_{y \in \calS_{t_{\out}}}\right)\right\}.  \label{eq:selection-good-subset}
\end{align}
By focusing on just the first term in the $\max$ in \eqref{eq:selection-good-subset}, we get
\begin{align*}
\min_{y \in \calS_{t_{\out}}} \ell_y \leq \ell_{y^{(1)}} + K \cdot \xi + 2\gamma\left(T, \frac{4\rho}{5}, \frac{4\beta}{5}\right).
\end{align*}
Then, from $\calE_3$, we have
\begin{align} \label{eq:final-util-bound}
\ell_{y_{\out}} = \min_{y \in \calS_{t_{\out}}} \ell_y \leq \ell_{y^{(1)}} + K \cdot \xi + 2\gamma\left(T, \frac{4\rho}{5}, \frac{4\beta}{5}\right) \overset{\text{(\Cref{lem:first-ineq})}}{\leq} \ell_{y^{(1)}} + \gamma(|\calY|, \rho, \beta),
\end{align}
as desired.

\paragraph{Bounding the Probability of Good Events.}
We have
\begin{align} \label{eq:good-prob-expand}
\Pr[\calE_1 \wedge \calE_2 \wedge \calE_3]
= \Pr[\calE_3 \mid \calE_1 \wedge \calE_2] \cdot \Pr[\calE_1 \wedge \calE_2]
\geq \Pr[\calE_3 \mid \calE_1 \wedge \calE_2] - \Pr[\neg \calE_1] - \Pr[\neg \calE_2].
\end{align}

We next bound each term separately.

\paragraph{Bounding $\Pr[\calE_3 \mid \calE_1 \wedge \calE_2]$.} Assume that both $\calE_1, \calE_2$ hold. 
By focusing on just the second term in the $\max$ in \eqref{eq:selection-good-subset}, we have 
\begin{align*}
\gap((\ell_y)_{y \in \calS_{t_{\out}}}) \geq \sqrt{K} \cdot \xi - 2\gamma\left(T, \frac{4\rho}{5}, \frac{4\beta}{5}\right) \overset{\text{(\Cref{lem:second-ineq})}}{\geq} \tau\left(|\calY|, \frac{\rho}{5}, \frac{\beta}{10}\right).
\end{align*}
Thus, by \Cref{lem:gap-bin-search}, $\Pr[\calE_3 \mid \calE_1 \wedge \calE_2] \geq 1 - \beta/10$.

\paragraph{Bounding $\Pr[\neg \calE_1]$.}
From \Cref{lem:prob-good-score-one-subset}, we have
\begin{align*}
\Pr[\neg \calE_1] \leq \left(1 - \frac{1}{2^{2\sqrt{K}}}\right)^T \leq  e^{-\frac{T}{2^{2\sqrt{K}}}} \leq e^{-2^{\sqrt{K}-1}} \leq \beta/10,
\end{align*}
where the penultimate inequality is from our choice of $T$ and the last inequality is from our assumptions $K \geq 1000$, which implies $\frac{2^{-K}}{10} \geq e^{-2^{\sqrt{K}-1}},$ and $\beta \geq 2^{-K}.$

\paragraph{Bounding $\Pr[\neg \calE_2]$.} It follows immediately from the inductive hypothesis that $\Pr[\neg \calE_2] \leq 4\beta/5$.

\paragraph{Putting Things Together.} Putting the three bounds together with \Cref{eq:good-prob-expand}, we have that $\Pr[\calE_1 \wedge \calE_2 \wedge \calE_3] \geq 1 - \beta$. As argued earlier, when $\calE_1, \calE_2, \calE_3$ all occur, \Cref{eq:final-util-bound} holds. This concludes our proof of the inductive step.

\subsection{From High-Probability to Expected Error: Proof of \Cref{thm:main-exp}}

Finally, we can turn the high-probability guarantee from \Cref{thm:main} to an expected error guarantee. The idea is simple: We run our algorithm and the binary tree algorithm. We then use a query to compare the two outputs and select the best one. This is described more formally in \Cref{alg:combined} and analyzed below.

\begin{algorithm}[h!]
\caption{\combined$\left((\ell_y)_{y \in \calY}; \rho\right)$}
\label{alg:combined}
\begin{algorithmic}[1]
\Require{Loss functions $(\ell_y)_{y \in \calY}$,~} \textbf{Parameter: } Privacy budget $\rho$
\Ensure{Index $y_{\out} \in \calY$}
\State $K \gets \lceil \log |\calY| \rceil$
\State $\beta \gets \frac{1}{K}$
\State $y^1_{\out} \gets \recurgap\left((\ell_y)_{y \in \calY}; \frac{\rho}{3}, \beta\right)$
\State $y^2_{\out} \gets \bintree\left((\ell_y)_{y \in \calY}; \frac{\rho}{3}\right)$
\State Query $q := \frac{1}{2}\left(\ell_{y^1_{\out}}(X) - \ell_{y^2_{\out}}(X)\right)$ with budget $\frac{\rho}{3}$; let the answer be $\tilde{q}$
\If{$\tq > 0$}
\State $y_{\out} \gets y^2_{\out}$
\Else
\State $y_{\out} \gets y^1_{\out}$
\EndIf
\Return $y_{\out}$
\end{algorithmic}
\end{algorithm}

\begin{proof}[Proof of \Cref{thm:main-exp}]
Let $Z = q - \tq$ denote the noise added to the query $q$. Recall that $Z \sim \calN\left(0, \frac{3}{2\rho}\right)$. Note that $\ell_{y_{\out}} - \min\left\{\ell_{y^1_{\out}}, \ell_{y^2_{\out}}\right\} \leq 2|Z|$. As a result, we have
\begin{align*}
\E[\ell_{y_{\out}}] \leq \E\left[\min\left\{\ell_{y^1_{\out}}, \ell_{y^2_{\out}}\right\}\right] + \E[2|Z|] \leq \E\left[\min\left\{\ell_{y^1_{\out}}, \ell_{y^2_{\out}}\right\}\right] + O\left(\frac{1}{\sqrt{\rho}}\right). 
\end{align*}
Next, let $\gamma = O\left(\frac{1}{\sqrt{\rho}} \cdot \log |\calY| \cdot (\log \log |\calY|)^{10} \cdot \log\left(\frac{\log |\calY|}{\beta}\right)\right)$ denote the error guarantee from \Cref{thm:main}. From this, consider $y'$ defined as follows:
\begin{align*}
y' =
\begin{cases}
y^1_{\out} & \text{ if } \ell_{y^1_{\out}} \leq \min_{y \in \calY} \ell_y + \gamma, \\
y^2_{\out} & \text{ otherwise.}
\end{cases}
\end{align*}
Clearly, $\ell_{y'} \geq \min\left\{\ell_{y^1_{\out}}, \ell_{y^2_{\out}}\right\}$. Also recall that $K = \lceil \log |\calY| \rceil$. Thus, we have
\begin{align}
&\E\left[\min\left\{\ell_{y^1_{\out}}, \ell_{y^2_{\out}}\right\}\right] \notag \\
&\leq \E\left[\ell_{y'}\right] \notag \\
&= \E\left[\ell_{y^{1}_{\out}} ~\middle\vert~ \ell_{y^1_{\out}} \leq \min_{y \in \calY} \ell_y + \gamma\right] \cdot \Pr\left[\ell_{y^1_{\out}} \leq \min_{y \in \calY} \ell_y + \gamma\right] \notag \\
&\qquad + \E\left[\ell_{y^{2}_{\out}} ~\middle\vert~ \ell_{y^1_{\out}} > \min_{y \in \calY} \ell_y + \gamma\right] \cdot \Pr\left[\ell_{y^1_{\out}} > \min_{y \in \calY} \ell_y + \gamma\right] \notag \\
&= \E\left[\ell_{y^{1}_{\out}} ~\middle\vert~ \ell_{y^1_{\out}} \leq \min_{y \in \calY} \ell_y + \gamma\right] \cdot \Pr\left[\ell_{y^1_{\out}} \leq \min_{y \in \calY} \ell_y + \gamma\right] + \E\left[\ell_{y^{2}_{\out}}\right] \cdot \Pr\left[\ell_{y^1_{\out}} > \min_{y \in \calY} \ell_y + \gamma\right] \label{eqn:indep}\\
&\leq \left(\min_{y \in \calY} \ell_y + \gamma\right) \cdot \Pr\left[\ell_{y^1_{\out}} \leq \min_{y \in \calY} \ell_y + \gamma\right] + \left(\min_{y \in \calY} \ell_y + O\left(\frac{K^{3/2}}{\sqrt{\rho}}\right)\right) \cdot \Pr\left[\ell_{y^1_{\out}} > \min_{y \in \calY} \ell_y + \gamma\right] \label{ineq:lemma4}\\
&\leq \min_{y \in \calY} \ell_y + \gamma + O\left(\frac{K^{3/2}}{\sqrt{\rho}}\right) \cdot \Pr\left[\ell_{y^1_{\out}} > \min_{y \in \calY} \ell_y + \gamma\right] \notag \\
&\leq \min_{y \in \calY} \ell_y + \gamma + O\left(\frac{K^{3/2}}{\sqrt{\rho}}\right) \cdot \beta \label{ineq:recurgap} \\
&\leq \min_{y \in \calY} \ell_y + O\left(\frac{1}{\sqrt{\rho}} \cdot K (\log K)^{11}\right), \label{ineq:usesK}
\end{align}
where \eqref{eqn:indep} is due to the fact that $y^2_{\out}$ is independent of the event $\ell_{y^1_{\out}} > \min_{y \in \calY} \ell_y + \gamma$, \eqref{ineq:lemma4} follows from \Cref{lem:bin-search-exp}, \eqref{ineq:recurgap} follows from the guarantee of~\recurgap~(\Cref{thm:main}), and \eqref{ineq:usesK} is due to our setting $\beta = \frac{1}{K}$ in Algorithm \eqref{alg:combined}. 

Thus, we can conclude that $\E[\ell_{y_{\out}}] \leq \min_{y \in \calY} \ell_y + O\left(\frac{1}{\sqrt{\rho}} \cdot K (\log K)^{11}\right)$ as desired.
\end{proof}



\subsection*{Acknowledgment}

We thank Charlie Harrison for introducing us to this problem and for subsequent insightful discussions. We also thank Thomas Steinke for helpful discussions.

\bibliography{ref}
\bibliographystyle{alpha}

\appendix

\section{On Equal vs Unequal Budgets}
\label{app:eqaul-v-unequal}

We use a model (\Cref{def:model-with-budget}) in which we are allowed to associate different budgets to the different queries. Strictly speaking, this is not the same as the model proposed in \cite{DPorg-open-problem-selection}, which uses the same budget for all queries. This is stated more precisely below.
\begin{definition}[Gaussian Mechanism with Equal Budget~\cite{DPorg-open-problem-selection}]
In the \emph{Gaussian mechanism with equal budget} model, the algorithm is given a finite budget $\rho > 0$ and can select the number of rounds $M$. At the $i$-th iteration, the algorithm can issue a sensitivity-1 query $q_i: \calX \to \R$, and it will receive back $q_i(X) + Z_i$ where $Z_i \sim \calN\left(0, \frac{M}{2\rho}\right)$. 
\end{definition}

Nevertheless, using the fact that the sum of Gaussian noises remain Gaussian, it is simple to show that the two models are equivalent, as formalized below. Note that the condition (that an upper bound on the number of rounds $M$ is known) is quite mild, and it can be easily seen that all algorithms discussed in this paper (e.g. \Cref{alg:final}) satisfy this condition.

\begin{lemma}
Let $\calA$ be any algorithm in the Gaussian mechanism with budget model that works in at most $M$ rounds and uses total budget $\rho$. Then, $\calA$ can be implemented in the Gaussian mechanism with equal budget model with total budget $2\rho$.
\end{lemma}

\begin{proof}
We simulate $\calA$ using $M' = 2M$ rounds in the Gaussian mechanism with equal budget model with total budget $\rho'  = 2\rho$ as follows:
\begin{itemize}
\item At round $i$ of $\calA$, suppose that $\calA$ issues a query $q_i$ with budget $\rho_i$.
\item We simulate it in the equal budget model by issuing $q_i$ a total of $m_i = \left\lceil \frac{M' \rho_i}{\rho'} \right\rceil$ times to get back answers $\tq_i^1, \dots, \tq^i_{m_i}$.
\item We then send back the answer $\tq_i = \frac{1}{m_i}\left(\tq_i^1 + \cdots + \tq^i_{m_i}\right) + \calN\left(0, \frac{1}{2\rho_i} - \frac{M'}{2\rho' m_i}\right)$.
\end{itemize}

It is simple to verify that the distribution of the answer $\tilde{q}_i$ is the same as that in the Gaussian mechanism with budget model, and that the total budget used is no more than $2\rho$ as desired.
\end{proof}

We note that our reduction does \emph{not} work for Laplace mechanism because the sum of Laplace noises does not follow the Laplace distribution.

\section{Proof of \Cref{lem:first-ineq,lem:second-ineq}}
\label{app:ineq-proofs}

To prove these lemmas, we will use the following proposition. Note that here and throughout, we allows the first arguments of $\xi, \gamma$ to be any positive real numbers.

\begin{proposition} \label{prop:cal-gemini}
For all $K \in \mathbb{R}$ such that $K \ge 1000$ and for all $\beta \in \mathbb{R}$ such that $0 < \beta \le 0.001$, the following inequality holds:
\[
\xi(K, \rho, \beta) \ge 36 \cdot \xi\left(3\sqrt{K}, \frac{4\rho}{5}, \frac{4\beta}{5}\right)
\]
\end{proposition}

\begin{proof}
We begin by substituting the definition of $\xi$ into the inequality. 
The inequality becomes:
\[
\frac{1000}{\sqrt{\rho}} (1 + \log_2 K)^{10} \log_2\left(\frac{1000(K+1)}{\beta}\right) \ge 18\sqrt{5} \cdot \frac{1000}{\sqrt{\rho}} \left(1 + \log_2(3\sqrt{K})\right)^{10} \log_2\left(\frac{1250(3\sqrt{K}+1)}{\beta}\right)
\]
Rearranging yields the equivalent inequality:
\[
\left(\frac{1 + \log_2 K}{1 + \log_2(3\sqrt{K})}\right)^{10} \cdot \frac{\log_2\left(\frac{250}{\beta} \cdot 4(K+1)\right)}{\log_2\left(\frac{250}{\beta} \cdot 5(3\sqrt{K}+1)\right)} \ge 18\sqrt{5}
\]
The proof proceeds by establishing lower bounds for the two multiplicative terms on the left.

\paragraph{Lower Bound for the First Term.}
Let us analyze the first term, which we denote $R_K(K)$:
\[
R_K(K) = \left(\frac{1 + \log_2 K}{1 + \log_2(3\sqrt{K})}\right)^{10} = \left(\frac{1 + \log_2 K}{1 + \log_2 3 + \frac{1}{2}\log_2 K}\right)^{10}
\]
Consider the base of this term as a function of $x = \log_2 K$. Let $g(x) = \frac{1+x}{1+\log_2 3 + 0.5x}$. The derivative with respect to $x$ is:
\[
g'(x) = \frac{(1)(1+\log_2 3 + 0.5x) - (1+x)(0.5)}{(1+\log_2 3 + 0.5x)^2} = \frac{0.5 + \log_2 3}{(1+\log_2 3 + 0.5x)^2} > 0,
\]
which implies that $g(x)$ is a strictly increasing function of $x$. For the domain $K \ge 1000$, the minimum value of $g(x)$ occurs at the minimum value of $x$, which is $x_{\min} = \log_2 1000$. We have $g(x) \geq g(x_{\min}) > 1.448$.
Thus, $R_K(K) > (1.448)^{10} > 40.5$.

\paragraph{Lower Bound for the Second Term.} We claim that the second term is at least one. This is because $4(K + 1) - 5(3\sqrt{K} + 1) = 3\sqrt{K}(\sqrt{K} - 5) + (K - 1) > 0$ under our assumption $K \geq 1000$.

\paragraph{Conclusion.}
By combining the lower bounds established in the preceding paragraphs, we have:
\[
\left(\frac{1 + \log_2 K}{1 + \log_2(3\sqrt{K})}\right)^{10} \cdot \frac{\log_2\left(\frac{250}{\beta} \cdot 4(K+1)\right)}{\log_2\left(\frac{250}{\beta} \cdot 5(3\sqrt{K}+1)\right)} > 40.5 > 18\sqrt{5}. \qedhere
\]
\end{proof}

\Cref{lem:first-ineq,lem:second-ineq} now follow rather simply from the above proposition.

\begin{proof}[Proof of \Cref{lem:first-ineq}]
Notice that $\xi$ is non-decreasing in its first argument and that $\lceil \log T \rceil \leq 3\sqrt{K}$. Thus, we have
\begin{align*}
K \cdot \xi\left(K, \rho, \beta\right) + 2\gamma\left(T, \frac{4\rho}{5}, \frac{4\beta}{5}\right) &= K \cdot \xi(K, \rho, \beta) + 2 \cdot 2\lceil \log T \rceil \cdot \xi\left(\lceil \log T \rceil, \frac{4\rho}{5}, \frac{4\beta}{5}\right) \\
&\leq K \cdot \xi(K, \rho, \beta) + 12 \sqrt{K} \cdot \xi\left(3\sqrt{K}, \frac{4\rho}{5}, \frac{4\beta}{5}\right) \\
(\text{\Cref{prop:cal-gemini}}) &\leq  \left(K + \frac{\sqrt{K}}{3}\right) \cdot \xi(K, \rho, \beta) \\
&\leq 2K \cdot \xi(K, \rho, \beta) \\
&= \gamma(|\calY|, \rho, \beta). \qedhere
\end{align*}
\end{proof}

\begin{proof}[Proof of \Cref{lem:second-ineq}]
Similar to the proof of \Cref{lem:first-ineq}, we have
\begin{align*}
\sqrt{K} \cdot \xi\left(K, \rho, \beta\right) - 2\gamma\left(T, \frac{4\rho}{5}, \frac{4\beta}{5}\right)
&= \sqrt{K} \cdot \xi\left(K, \rho, \beta\right) - 2 \cdot 2\lceil \log T \rceil \cdot \xi\left(\lceil \log T \rceil, \frac{4\rho}{5}, \frac{4\beta}{5}\right) \\
&\geq \sqrt{K} \cdot \xi\left(K, \rho, \beta\right) - 12 \sqrt{K} \cdot \xi\left(3\sqrt{K}, \frac{4\rho}{5}, \frac{4\beta}{5}\right) \\
(\text{\Cref{prop:cal-gemini}}) &\geq \frac{\sqrt{K}}{2} \cdot \xi\left(K, \rho, \beta\right) \\
&= \frac{\sqrt{K}}{2} \cdot \frac{1000}{\sqrt{\rho}} (1 + \log K)^{10} \log(1000(K+1)/\beta) \\
&\geq \frac{\sqrt{K}}{2} \cdot \frac{4\sqrt{5}}{\sqrt{\rho}} \cdot 1 \cdot \sqrt{\log(10K/\beta)} \\
&= \tau\left(|\calY|, \frac{\rho}{5}, \frac{\beta}{10}\right)  \qedhere
\end{align*}
\end{proof}
\end{document}

%% file: preamble.tex
\pdfoutput=1
\usepackage[T1]{fontenc}
\usepackage{lmodern}
\usepackage{slantsc}
\usepackage[protrusion=true,expansion=true]{microtype}
\usepackage{breakcites}
\usepackage{amsmath,amssymb,amsfonts,amsthm}
\usepackage{subcaption}
\usepackage{graphicx}
\usepackage{fullpage}
\usepackage{setspace}
\usepackage[backref=page]{hyperref}
\usepackage{color}
\usepackage{wrapfig}
\usepackage{tikz}
\usetikzlibrary{decorations.pathreplacing}
\usepackage{algorithm}
\usepackage[noend]{algpseudocode}
\usepackage[framemethod=tikz]{mdframed}
\usepackage{xspace}
\usepackage{pgfplots}
\usepackage{framed}
\usepackage{thmtools}
\usepackage{thm-restate}
\usepackage{tabu}
\usepackage{fancyhdr}
\usepackage{cleveref}
\pgfplotsset{compat=1.5}

\newtheorem{theorem}{Theorem}

\newtheorem{lemma}[theorem]{Lemma}
\newtheorem{proposition}[theorem]{Proposition}
\newtheorem{definition}[theorem]{Definition}

\newenvironment{proofof}[1]{\begin{trivlist} \item {\bf Proof
#1:~~}}
  {\qed\end{trivlist}}

\newcommand{\namedref}[2]{\hyperref[#2]{#1~\ref*{#2}}}

\renewcommand{\O}[1]{\ensuremath{O\left(#1\right)}}
\newcommand{\tO}[1]{\ensuremath{\tilde{O}\left(#1\right)}}
\newcommand{\eps}{\varepsilon}

\def \calA    {\mdef{\mathcal{A}}}

\def \calC    {\mdef{\mathcal{C}}}

\def \calE    {\mdef{\mathcal{E}}}

\def \calN    {\mdef{\mathcal{N}}}

\def \calS    {\mdef{\mathcal{S}}}

\def \calX    {\mdef{\mathcal{X}}}
\def \calY    {\mdef{\mathcal{Y}}}

\def \bX    {\mdef{\mathbf{X}}}

\def \R    {\mdef{\mathbb{R}}}

\newcommand{\mdef}[1]{{\ensuremath{#1}}\xspace}  

\DeclareMathOperator*{\argmin}{argmin}

\newcommand{\E}[2][]{\mdef{\underset{#1}{\mathbb{E}}\left[#2\right]}} 
\renewcommand{\E}{\mathbb{E}}
\newcommand{\ignore}[1]{}

\newif\ifnotes\notestrue 
\ifnotes
\newcommand{\samson}[1]{\textcolor{blue}{{\bf (Samson:} {#1}{\bf ) }} \marginpar{\tiny\bf
             \begin{minipage}[t]{0.5in}
               \raggedright S:
            \end{minipage}}}

\newcommand{\david}[1]{\textcolor{purple}{{\bf (David:} {#1}{\bf ) }} \marginpar{\tiny\bf
             \begin{minipage}[t]{0.5in}
               \raggedright D:
            \end{minipage}}} 
\else
\newcommand{\samson}[1]{}
\newcommand{\david}[1]{}
\fi

\makeatletter
\renewcommand*{\@fnsymbol}[1]{\textcolor{mahogany}{\ensuremath{\ifcase#1\or *\or \dagger\or \ddagger\or
 \mathsection\or \triangledown\or \mathparagraph\or \|\or **\or \dagger\dagger
   \or \ddagger\ddagger \else\@ctrerr\fi}}}
\makeatother

\providecommand{\email}[1]{\href{mailto:#1}{\nolinkurl{#1}\xspace}}

\definecolor{mahogany}{rgb}{0.75, 0.25, 0.0}
\definecolor{darkblue}{rgb}{0.0, 0.0, 0.55}
\definecolor{darkpastelgreen}{rgb}{0.01, 0.75, 0.24}
\definecolor{darkgreen}{rgb}{0.0, 0.2, 0.13}
\definecolor{darkgoldenrod}{rgb}{0.72, 0.53, 0.04}
\definecolor{darkred}{rgb}{0.55, 0.0, 0.0}
\definecolor{forestgreenweb}{rgb}{0.13, 0.55, 0.13}
\definecolor{greencss}{rgb}{0.0, 0.5, 0.0}
\definecolor{bleudefrance}{rgb}{0.19, 0.55, 0.91}
\definecolor{darkcerulean}{rgb}{0.03, 0.27, 0.49}

\hypersetup{
     colorlinks   = true,
     citecolor    = darkcerulean,
	 linkcolor	  = darkcerulean
}

\fancypagestyle{pg}
{
\lhead{}
\rhead{}
\cfoot{--\ \thepage\ --}

}

\AtBeginDocument{%
  \DeclareFontShape{T1}{lmr}{m}{scit}{<->ssub*lmr/m/scsl}{}%
}

\newcommand{\tell}{\tilde{\ell}}

\newcommand{\gap}{\mathrm{gap}}

\newcommand{\out}{\mathrm{out}}
\newcommand{\tq}{\tilde{q}}
\newcommand{\bintree}{\textsc{BinTree}}
\newcommand{\recurgap}{\textsc{RecurGap}}
\newcommand{\combined}{\textsc{Combined}}

%% file: main.bbl
\newcommand{\etalchar}[1]{$^{#1}$}
\begin{thebibliography}{WRRW23}

\bibitem[BNS16]{BeimelNS16}
Amos Beimel, Kobbi Nissim, and Uri Stemmer.
\newblock Private learning and sanitization: Pure vs. approximate differential
  privacy.
\newblock {\em Theory Comput.}, 12(1):1--61, 2016.

\bibitem[BS16]{BunS16}
Mark Bun and Thomas Steinke.
\newblock Concentrated differential privacy: Simplifications, extensions, and
  lower bounds.
\newblock In {\em TCC}, page 635–658, 2016.

\bibitem[BST14]{BassilyST14}
Raef Bassily, Adam~D. Smith, and Abhradeep Thakurta.
\newblock Private empirical risk minimization: Efficient algorithms and tight
  error bounds.
\newblock In {\em FOCS}, pages 464--473, 2014.

\bibitem[DKM{\etalchar{+}}06]{DworkKMMN06}
Cynthia Dwork, Krishnaram Kenthapadi, Frank McSherry, Ilya Mironov, and Moni
  Naor.
\newblock Our data, ourselves: Privacy via distributed noise generation.
\newblock In {\em EUROCRYPT}, pages 486--503, 2006.

\bibitem[DMNS06]{dwork-calibrating}
Cynthia Dwork, Frank McSherry, Kobbi Nissim, and Adam Smith.
\newblock Calibrating noise to sensitivity in private data analysis.
\newblock In {\em TCC}, page 265–284, 2006.

\bibitem[DR16]{DworkR16}
Cynthia Dwork and Guy~N. Rothblum.
\newblock Concentrated differential privacy.
\newblock {\em CoRR}, abs/1603.01887, 2016.

\bibitem[GLM{\etalchar{+}}10]{GuptaLMRT10}
Anupam Gupta, Katrina Ligett, Frank McSherry, Aaron Roth, and Kunal Talwar.
\newblock Differentially private combinatorial optimization.
\newblock In {\em SODA}, pages 1106--1125, 2010.

\bibitem[MT07]{McSherryT07}
Frank McSherry and Kunal Talwar.
\newblock Mechanism design via differential privacy.
\newblock In {\em FOCS}, pages 94--103, 2007.

\bibitem[Ste25]{DPorg-open-problem-selection}
Thomas Steinke.
\newblock Open problem: Selection via low-sensitivity queries.
\newblock DifferentialPrivacy.org, 05 2025.
\newblock \url{https://differentialprivacy.org/open-problem-selection/}.

\bibitem[WRRW23]{WhitehouseR0023}
Justin Whitehouse, Aaditya Ramdas, Ryan Rogers, and Steven Wu.
\newblock Fully-adaptive composition in differential privacy.
\newblock In {\em ICML}, pages 36990--37007, 2023.

\end{thebibliography}
